\documentclass[twoside,a4paper]{article}

\usepackage{amsmath}
\usepackage{amsthm} 
\usepackage{amssymb}

\usepackage{txfonts}
\usepackage[utf8]{inputenc}
\usepackage[T1]{fontenc}
\usepackage{lmodern}
\usepackage{hyperref} 

\newcommand{\CH[2]}{\textbf{HC}_{#1}\left (#2 \right)}
\newcommand{\CHH[3]}{\textbf{HC}_{#1}^{#2}\left (#3 \right )}
\newcommand{\VP}{\textsc{VP}}
\newcommand{\VNP}{\textsc{VNP}}
\newcommand{\cut}{\text{Cut}^2}
\newcommand{\Cut[1]}{\text{Cut}^{#1}}
\newcommand{\p}{\textsc{P}}
\newcommand{\NP}{\textsc{NP}}
\newcommand{\VAC}{\textsc{VAC}_0}
\newcommand{\per}{\textsc{per}}
\renewcommand{\multimap}{\multimapdotboth}

\newtheorem{thm}{Theorem}
\newtheorem*{thm*}{Theorem}
\newtheorem{lem}{Lemma}
\newtheorem*{lem*}{Lemma}
\newtheorem{cor}{Corollary}
\newtheorem{prop}{Proposition}

\theoremstyle{definition}
\newtheorem{defi}{Definition}
\theoremstyle{remarks}
\newtheorem*{rem}{Remark}

\begin{document}
\title{A Dichotomy Theorem for Homomorphism Polynomials}
\author{N. de Rugy-Altherre \\Univ Paris Diderot, Sorbonne Paris Cité,
\\Institut de Mathématiques de Jussieu, UMR 7586 CNRS, \\F-75205 Paris, France\\
nderugy@math.univ-paris-diderot.fr}

\maketitle

\abstract
In the present paper we show a dichotomy theorem for the complexity of polynomial evaluation. We associate to each graph $H$ a polynomial that encodes all graphs of a fixed size homomorphic to $H$. We show that this family is computable by arithmetic circuits in constant depth if  $H$ has a loop or no edge and that it is hard  otherwise (i.e., complete for $\VNP$, the arithmetic class related to $\#P$). We also demonstrate the hardness over $\mathbb Q$ of cut eliminator, a polynomial defined by Bürgisser which is known to be neither $\VP$ nor $\VNP$-complete in $\mathbb F_2$, if $\VP \neq \VNP$ ($\VP$ is the class of polynomials computable by arithmetic circuit of polynomial size).

\section{Introduction}
Dichotomy theorems are a way to classify entire classes of problems by their complexity. For instance, Schaefer~\cite{Sch78} has classified a subclass of CSP as being in $\p$ or $\NP$-complete. Whether such a theorem holds for every CSP is a famous open question, the Dichotomy Conjecture~\cite{FV98}. More recently, dichotomy theorems for counting homomorphism problems have been studied by several authors~\cite{DG00,BG04,GT11}. In Valiant's theory of polynomial evaluation, an equivalent of Schaefer's theorem has been produced by Briquel and Koiran ~\cite{BK09}, opening the way for such theorems in this theory. 

In this paper, we define the class of \emph{homomorphism polynomials}: if $H$ is a graph, the monomials of the associated homomorphism polynomial $f^H_n$ encode the list of graphs of size $n$ homomorphic to $H$. Our main result is a dichotomy theorem for this class of polynomials, in Valiant's theory : $(f^H_n)$ is computable by unbounded fan-in arithmetic circuits in constant depth (the Valiant equivalent of AC$_0$) if  $H$ has a loop or no edge and hard  otherwise (i.e., complete for $\VNP$, the arithmetic class related to $\#P$).

Deciding the existence of or counting homomorphisms are important problems which generalize many natural questions ($k$-coloring, acyclicity, binary CSP, etc~\cite{HN04}). The weighted version, which is very similar to our polynomials, has several applications in statistical physics~\cite{Wel93}. 

The connection between Valiant's theory of polynomial evaluation and counting problems is well known. For instance,  the permanent is complete in both settings~\cite{Val791,Val792}. Homomorphism polynomials are therefore linked to counting homomorphism problems: the evaluation of $f^H_n$ counts the number of graphs of size $n$ homomorphic to $H$. But these polynomials also tell us something about enumeration: it is equivalent to enumerate this set of graphs and to enumerate all the monomials of $f^H_n$. We will use this property to demonstrate a result on representations by first order formulas of the bipartite property for graphs.

The hard part of our dichotomy theorem is to demonstrate the hardness of most homomorphism polynomials. We therefore obtain many  new $\VNP$-complete families (indeed unlike $\NP$-complete problems, there are not a lot of $\VNP$-complete polynomials known). These families can be seen as generating functions in the way introduced by Jerrum and Bürgisser~\cite{Bur00}. No general hardness theorem has been found for these functions (but there are large results for tractable functions~\cite{CMR01}). As stated before, many problems can be seen as homomorphism problems. Our theorem therefore implies the hardness of many generating functions.

Moreover, Bürgisser~\cite{Bur00} gives a family of polynomials, $\cut$, which is neither $\VNP$-complete nor $\VP$ over $\mathbb F_2$, if $\VNP \neq \VP$ of course. He also demonstrates the existence of such a family in every field, but does not give a specific family for $\mathbb Q$. He wonder whether $\cut$ can be such a family in $\mathbb Q$. During the demonstration of our theorem, we will show the $\VNP$-completeness of $\cut$ and therefore answer negatively the question.

\section{Definitions and general discussions}

A polynomial enumerates a family $\mathcal P$ of graphs if each monomial of this polynomial encodes a single graph $G$ in $\mathcal P$ and each graph $G \in \mathcal P$ is encoded by a single monomial. In this paper we study polynomials enumerating graphs $G$ homomorphic to a graph $H$ using a simple code: the variables of the polynomial are the set $\{x_e, e \in E(K_n)\}$; the monomial coding $G$ is $\prod_{e \in E(G)} x^e$.  These polynomials are therefore multilinear and have only $0$ or $1$ for coefficients. More precisely

\begin{defi}
Let $H$ be a graph. The polynomial enumerating all graphs $G$ with $n$ vertices which are homomorphic to $H$ is:
	\[f^H_n((x_e)_{e \in E(K_n)}) := \sum_{\bar \epsilon \in \{0,1\}^{|E(K_n)|}} \Phi_n^H(\bar \epsilon) \prod_{e \in E(K_n)} x_e^{\epsilon_n}\]
	Where $\Phi_n^H(\bar \epsilon)$ is equal to $1$ if the graph $G$ such that $V(G) = [n]$ and $E(G) = \{e \in E(K_n) | \epsilon_e = 1 \}$ is homomorphic to $H$, $0$ otherwise.
\end{defi}

We work within Valiant's algebraic framework. Here is a brief introduction to this complexity theory. For a more complete overview, see~\cite{Bur00}.

An \textit{arithmetic circuit} over a field $\mathbb K$ is a labeled directed acyclic connected graph with vertices of indegree $0$ or $2$ and only one sink. The vertices with indegree $0$ are called \textit{input gates} and are labeled with variables or constants from $\mathbb K$. The vertices with indegree $2$ are called \textit{computation gates} and are labeled with $\times$ or $+$. The sink of the circuit is called the \textit{output gate}.

The polynomial computed by a gate of the arithmetic circuit is defined by induction : an input gate computes its label; a computation gate computes the product or the sum of its children's' values. The polynomial computed by an arithmetic circuit is the polynomial computed by the sink of the circuit.

A \textit{p-family} is a sequence $(f_n)$ of polynomials over a field $\mathbb K$ such that the number of variables as well as the degree of $f_n$ is polynomially bounded in $n$. The \textit{complexity} $L(f)$ of a polynomial $f \in \mathbb K[x_1, \hdots, x_n]$ is the minimal number of computational gates of an arithmetic circuit computing $f$ from variables $x_1, \hdots, x_n$ and constants in $\mathbb K$. 

Two of the main classes in this theory are: the analog of $\p$, $\VP$, which contains of every p-family $(f_n)$ such that $L(f_n)$ is a function polynomially bounded in $n$; and the analog of $\NP$, $\VNP$. A p-family $(f_n)$ is in $\VNP$ iff there exists a $\VP$ family $(g_n)$ such that for all $n$, $f_n(x_1,\hdots, x_n) = \sum_{\bar{\epsilon} \in \{0,1\}^n} g_n(x_1,\hdots, x_n, \epsilon_1, \hdots, \epsilon_n)$.

Furthermore we introduce the Valiant equivalent of AC$_0$, $\VAC$ as the class of every p-family $(f_n)$ such that there exists a circuit of constant depth but with unfounded fan-in for the computation gates computing $f_n$ and of size polynomially bounded in $n$. It has been introduce in~\cite{MR09}.

As in most complexity theory we have a natural notion of reduction: a polynomial $f$ is a projection of a polynomial $g$, write $f \leq_p g$, if there are values $a_i \in \mathbb K \cup \{x_1, \hdots x_n\}$ such that $f(\bar x) = g(\bar a)$. A p-family $(f_n)$ is a p-projection of $(g_n)$ if there exists a polynomially bounded function $p$ such that for every $n$, $f_n \leq_p g_{p(n)}$. As usual we say that $(g_n)$ is complete for a class $C$ if, for every $(f_n) \in C$, $(f_n) \leq_p (g_n)$ and if $(g_n) \in C$.

The three classes presented above are closed under projections. We will need in this paper a second notion of reduction, the c-reduction: the oracle complexity $L^g(f)$ of a polynomial $f$ with the knowledge of $g$ is the minimum number of computation gates and evaluations of $g$ over previously computed values that are sufficient to compute $f$ from the variables $x_1, \hdots x_n$ and constants from $\mathbb K$. A p-family $(f_n)$ \textit{c-reduces} to $(g_n)$ if there exists a polynomially bounded function $p$ such that $L^{g_{p(n)}}(f_n)$ is a polynomially bounded function. 

This reduction is more powerful than the projection: if $(f_n) \leq_p (g_n)$, then $(f_n) \leq_c (g_n)$. With this reduction, $\VNP$ is still closed, but it is harder to demonstrate (See~\cite{Poi08} for a idea of the proof). However this reduction does not distinguish the lower classes. For example, $0$ is $\VP$-complete for c-reductions. 

The notion of c-reductions was introduced by Bürgisser~\cite{Bur00} in his work on p-families that are neither $\VP$ nor $\VNP$-complete (if $\VNP \neq \VP$). It has been rarely used to demonstrate the hardness of computing polynomials, though there is a recent example in~\cite{Bri11}. 

Bürgisser demonstrates the existence of such a family in every field with an abstract embedding theorem. Furthermore he gives a specific family for some finite field, the \textit{cut enumerator} $\Cut[q]_n$:
\[\Cut[q]_n := \sum_S \prod_{i \in A,j \in B} x_{ij}^{q-1}\]
where the sum is over all cuts $S=\{A,B\}$ of the complete graph $K_n$ on the set of nodes $\{1,\hdots,n\}$. (A cut of a graph is a partition of its set of nodes into two nonempty subsets).

\begin{thm}(Bürgisser, 5.22)\label{intro_Burgisser}
	Let $q$ be a power of the prime $p$. The family of cut enumerators $\Cut[q]$ over a finite field $\mathbb{F}_q$ is neither $\VP$ nor $\VNP$-complete with respect to c-reductions, provided $\text{Mod}_p \textsc{NP} \nsubseteq \textsc{P} / \text{poly}$. The latter condition is satisfied if the polynomial hierarchy does not collapse at the second level. 
\end{thm}

Howerver, the cut enumerator can be rewritten in the following way :
\[\Cut[q]_n(\bar x) = \sum_{\substack{V \subsetneq [n] \\ V \neq \emptyset}} \prod_{i \in V} \prod_{j \in V^c} x_{i,j}^{q-1}\]

This polynomial will be studied in Lemma~\ref{bip-complet-oriente}, for other purposes. In particular, we will demonstrate that this family is $\VNP$-complete in $\mathbb Q$ if $q=2$, which answers a question of Bürgisser. Furthermore, this implies that our demonstration does not work in  a field of characteristic $2$ (if $\VP \neq \VNP$).

The main theorem of our paper is:
\begin{thm}\label{theoreme}
Let $H$ be a graph. Let $f_n^H$ be the polynomial enumerating all graphs homomorphic to $H$. Then, over $\mathbb Q$
\begin{itemize}
	\item If $H$ has a loop or no edges, $(f_n^H)$ is in $\VAC$.
	\item Else $(f_n^H)$ is VNP-complete for c-reductions.
\end{itemize}	 
\end{thm}

Proofs that all our p-families are p-families and are in $\VNP$ are left to the reader. It is a simple application of Valiant's criterion~\cite{Bur00}. We will demonstrate our main theorem in four steps: first we will deal with the easy cases. Secondly, using Dyer and Greenhill's ideas, we will reduce the problem to the case where $H$ is as simple as possible, i.e., with only one edge and two vertices, $\multimap$. In a third step, we will note that the monomial of $f^{\multimap}$ have a kind of ``hereditar'' property which we will use to reduce the problem to a simpler polynomial. Finally, we will demonstrate the $\VNP$-completeness of this simpler polynomial.

As application of our theorem and using result on enumeration, we can prove that:
\begin{cor}\label{enumeration}
If $\VP \neq \VNP$, and if the product of two boolean matrices of size $n$ cannot be computed in $\mathcal O(n^2)$, being bipartite cannot be expressed by an acyclic conjunctive first order formula of polynomial size on the size of the studied graph.
\end{cor}

We just give here an idea of the demonstration. Let $\phi_n$ be an acylic conjunctive first order formula of polynomial size stating that $G$ is bipartite or not, for every graph $G$ of size $ n$. As the set of bipartite graphs can be enumerated with a constant delay, by the dichotomy theorem of  Bagan~\cite{Bag09}, $\phi_n$ is CCQ, i.e., of star size $1$. Therefore, thanks to the work of Durand and Mengel~\cite{DM11}, we know that $P(\phi) \in \VP$. At last, as $P(\phi) = f^{\multimap}$, $f^{\multimap}$ is $\VNP$-complete and in $\VP$, which is impossible if $\VNP \neq \VP$. Therefore $\phi_n$ cannot exist.

\section{The easy cases}

\begin{lem}\label{easy_case_1}
If $H$ has a loop, then $(f_n^H)$ is in $\VAC$.
\end{lem}

\begin{proof}
	Let $v$ be a looped vertex of $H$. Then every graph $G$ is homomorphic to $ H $: we just have to send all vertices of $ G $ to $ v $. Therefore,
\[f^H_n(\bar x) = \sum_{\bar \epsilon \in \{0,1\}^{|E(K_n)|}}  \prod_{e \in E(K_n)} x_e^{\epsilon_n} = \prod_{e \in E(K_n)} \left (1 + x_e \right) \]
This product can be computed easily by an arithmetic circuit of polynomial size, depth $3$ and unbounded fan-in for the $\times$-gates.
\end{proof}

The case where $ H $ has no edges is quite obvious, we just have to notice that if $H_1$ and $H_2$ are two  bihomomorphic graphs, i.e., there exists a homomorphism from $H_1$ to $H_2$ and vice versa, then, for all $n$, $f_n^{H_1} = f_n^{H_2}$. Let us consider now our graph $H$ with no edges. It is bihomomorphic to the graph $H_0$, composed of a single vertex. Furthermore the only graphs homomorphic to $H_0$ are those with no edges. Therefore $f_n^{H_0} = 1$.

Let us quote a very useful technical result, already known in~\cite{Bur00}.

\begin{lem}\label{compo-homo}
Let $(f_n)$ be a p-family. Let us write $\CH[k]{f_n}$ for the homogeneous component of degree $k$ of $f_n$.

Then for any sequence of integers $(k_n) $ there exists a c-reduction from the homogeneous component to the polynomial itself:
	\[\left (\CH[k_n]{f_n} \right) \leq_c \left(f_n\right)\]
\end{lem}

\begin{proof}
Let $n$ be an integer and $d_n$ be the degree of $f_n$. Let us recall that the degree of $f_n$ is polynomially bounded in $n$. To begin, let us notice that for all $i$, $f_n(2^i\bar x)  = \sum_{k=0}^{d_n} (2^i)^k \CH[k]{f_n}(\bar x) $. If we write $c_k = \CH[k]{f_n}(\bar x)$, $f_i = f_n(2^i \bar x)$ and $w_k = 2^k$, those equations can be written as:

\[
\begin{pmatrix}
   	f_1 \\
	f_2 \\
   	\vdots \\
	f_{d_n} \\ 
\end{pmatrix} = 
\begin{pmatrix}
   	1 & 1 & \cdots & 1 \\
   	\omega_1 & \omega_2 & \cdots & \omega_d \\
   	\vdots & \vdots & \vdots & \vdots \\ 
   	\omega_1^{d_n} & \omega_2^d & \cdots & \omega_{d_n}^{d_n}
\end{pmatrix}
\begin{pmatrix}
   	c_1 \\
	c_2 \\
   	\vdots \\
	c_{d_n} \\ 
\end{pmatrix}
\]

Let us write $ V $ the square matrix of this equation. It is a Vandermonde matrix with non negative integers. As the coefficients $w_k$ are  distinct, this matrix is invertible. Therefore there exist some rationals $w_{i,k}^*$ such that
$c_k = \sum_{i=0}^{d_n} w_{i,k}^*f^n(2^i\bar x)$.

 So we have, for all $n$ and all $k \leq {d_n}$, $L^{f_n}(\CH[k]{f_n}) = \mathcal{O}(n \times {d_n})$. Let $(k_n)$ be a sequence of integers. If $k_n > d_n$, then $\CH[k_n]{f_n} = 0$. Thus for all $L^{f_n}(\CH[k_n]{f_n}) \leq \mathcal{O}(n \times d_n)$. As $ d_n $ is polynomially bounded by $n$, we have just constructed the c-reduction we were looking for.
\end{proof}

This lemma is still true if we take the homogeneous compenent in some variables only. Let $f(\bar x, \bar y)$ be a polynomial and $\CHH[k]{\bar y}{f}$ the homogeneous compenent of size $k$ in variables $\bar y$, i.e., considering $\bar x$ as constant. Then for any p-family $(f_n)$ and any sequence of integers $k_n$,$\left (\CHH[k]{\bar y}{f_n} \right) \leq_c \left(f_n \right)$.

\begin{rem}
If $C$ is a circuit that compute of size $s$ a polynomial $f$ of degree $d$, we know that there exists an arithmetic circuit of size $d^2s$ that compute all homogeneous component of $f$ (Theorem 2.2 in~\cite{SY10} for example). With the previous lemma we can see that if $C$ is a arithmetic formula of size $s$ that compute $f$ of degree $d$, then for any $k \leq d$, we can built an arithmetic formula of size $\mathcal O(d^2s)$ that compute the homogeneous component of degree $k$ of $f$ (this formula must use constant of $\mathbb Q$).
\end{rem}

\section{Reduction to bipartite graphs}
One of the main idea of Dyer and Greenhill is to reduce the problem of counting homomorphisms from $G$ to $H$ to a simpler graph $H$. Here we have a similar result but the graph we obtain is extremely simple.
\begin{prop}\label{reduction-aux-biparties}
	Let $\multimap$ be the graph with two vertices and one edge. Then for any graph $H$ with no loops and at least one edge:
\[f_n^{\multimap}  \leq_c f_n^H\]
\end{prop}

\begin{lem}\label{reduction-de-H-a-H-prime}
	Let $H=(V,E)$ be a graph with no loops and at least one edge. For all $i \in V $, let us call $H_i$ the subgraph of $H$ induced by the neighbours of $i$ and $H'$ the disjoint union of all those neighbourhood graphs. Then
\[f_n^{H'} \leq_c f_n^H\] 
\end{lem}

\begin{proof}
	For any graph $G$, let us call $G'$ the graph built from $G$ by adding a new vertex $v$ and joining it to every vertex of $G$. Dyer and Greenhill have noticed in their paper~\cite{DG00} that the number of homomorphisms from $G'$ to $H$ is equal to the number of homomorphisms from $G$ to $H'$. Let us see what this means for our polynomials.

Let $\tilde{f}^H_n(\bar x)$ be the polynomial built from $f^H_{n+1}$ by replacing all $x_{i,n+1}$ variables by $y$, by taking the homogeneous component of degree $n$ in $y$ and by giving value $1$ to $y$: $\tilde{f}^H_n(\bar x) = \CHH[n]{y}{f_{n+1}^H(\bar x,y)}_{|_{y=1}}$.

$\tilde{f}^H_n$ enumerates all graphs of size $n+1$ homomorphic to $H$ and having a vertex $v$ linked to all others, i.e.,  graphs $G'$ homomorphic to $H$.
\[\tilde{f}^H_n(\bar x)  = \sum_{G \text{ graph of size }n}\Phi^H_n(G') \bar x^{\bar G} 
	 = \sum_{G \text{ graph of size }n}\Phi^{H'}_n(G) \bar x^{\bar G} 
	 := f^{H'}_n(\bar x) \]

where $\Phi^H_n(G) = 1$ if there exists a homomorphism between $G$ and $H$, $0$ otherwise. Then $\tilde{f}^H_n = f_n^{H'}$ and consequently there is a c-reduction from $(f^{H'}_n)$ to $(f_{n+1}^H)$.

\end{proof}

\begin{proof}[Proof of the proposition~\ref{reduction-aux-biparties}]

Since $H$ has no loops, the degree of $H'$ is strictly lower than that of $H$. Indeed, for any vertex $j$ of $H$, any vertex $k$ linked to $j$ in $H'$ is also linked to $j$ in $H$. So the degree of $H'$ cannot grow. Furthermore, if $i$ is a maximal degree ($d$) vertex of $H$ and if $j$ is a vertex of $H$, then either $i$ is not linked to $j$ and so does not appear in $H_j$, or it is linked to $j$ and, in $H_j$, it cannot be linked to more than $d-1$ vertices ($|V_i \cap V_j|$, knowing that $j \notin V_j$). The degree of the vertices of $H'$ corresponding to the vertex $i$ is therefore strictly  lower than $d$.

If we apply the process of Lemma~\ref{reduction-de-H-a-H-prime} to $H$ a finite number of times we obtain a $1$-regular graph. This graph is bihomomorphic to $\multimap$. Therefore we have the expected reduction: $f_n^{\multimap}  \leq_c f^H_n$
\end{proof}

\section{The bipartite case}
Now we have to demonstrate the following proposition.
\begin{prop}
Let $\multimap$ be the graph with two vertices and one edge. Then $f_n^{\multimap} $ is VNP-complete.
\end{prop}

\subsection{Hereditary polynomials}\label{hereditarity}
$f_n^{\multimap}$ has a noteworthy property: if $(x_i)_{i \in I}$ is a monomial of $f_n^{\multimap} $, then for all $J \subset I$, $(x_j)_{j \in J}$ is also a monomial of $f_n^{\multimap} $. Indeed, this polynomial enumerates all graphs homomorphic to $\multimap$, i.e., all bipartite graphs, and every subgraph of a graph homomorphic to $\multimap$ is also homomorphic to $\multimap$. To build on this idea, we introduce the following definitions.

Let $f$ be a multilinear polynomial with $0$,$1$ coefficients only. $f$ is \textit{hereditary} if it satisfies the following property:  if $(x_i)_{i \in I}$ is a monomial of $f$, so is $(x_j)_{j \in J}$, for all $J \subset I$. This relation induces an order on all the monomials of $f$: $(x_i)_{i \in I}$ is greater that $(x_j)_{j \in J}$ if $J \subset I$. We call a monomial a \textit{generator} for $f$ if it is maximal for this order.

A hereditary polynomial is completely defined by its generators. We can naturally ask the opposite question. Let $g$ be a multilinear polynomial with $0$,$1$ coefficients. We write $\downarrow g$ for the \textit{son} of $g$, i.e., the smallest hereditary polynomial containing $g$.

\begin{lem}
	Let $g_n$ be a sequence of multilinear homogeneous polynomials with $0$,$1$ coefficients. Then 
\[(g_n) \leq_c (\downarrow g_n) \]
\end{lem}

\begin{proof}
	Let $n$ be a integer and $d$ the degree of all the monomials of $g_n$. The generators of $\downarrow g_n$ are simply the monomials of degree $d$, $g_n = \CH[d]{\downarrow g_n}$. The reduction is given by Lemma~\ref{compo-homo}.
\end{proof}

The reduction $(g_n) \leq_c (\downarrow g_n) $ is probably not true in the general case: for instance $f_n = (\prod_{i=1}^n x_i) + per(\bar x) $ is $\VNP$-complete, but its son is the sum of all monomials of size less than $n$, and it is therefore in $\VAC$.  However, we can hope to generalize this equality to all \textit{pure} polynomials, which we define as polynomials with no pair of comparable monomials.

In our case, we try to demonstrate that $f_n^{\multimap} $ is VNP-complete. This polynomial is hereditary. In Lemma~\ref{bip-complet}, we will demonstrate that the polynomial $F_n$, which enumerates all graphs composed of some isolated vertices and a complete bipartite connected component, is VNP-complete. Now $f_n^{\multimap} $ is the son of this polynomial: $\downarrow F_n = f_n^{\multimap} $.

Unfortunately this polynomial $F_n$ is not pure, some of its monomials are comparable. In Lemma~\ref{bip-complet-max} we demonstrate that the polynomial $(G_n)$, enumerating all the complete bipartite graphs, is also VNP-complete. $f_n^{\multimap} $ is also the son of $G_n$, but  this polynomial is pure!

As we have no general proof of the reduction from a pure polynomial to its son (yet), we will now give an ad hoc demonstration of the reduction from $G_n$ to $f_n^{\multimap} $.

At the end, we will have demonstrated the following chain, for any graph $H$ with no loop and at least one edge.
	\[ \left (  \text{GF}(K_n, \text{clique})\right )\leq_c (F_n)\leq_p (\cut_n)\leq_c (G_n) \leq_c (f^{\multimap}_n) \leq_c (f^H_n)\]

As the generating function of clique is $\VNP$ complete, we will have demonstrated Theorem~\ref{theoreme}

\begin{lem}\label{reduction_des_bip_aux_bip_complet}
Let $f_n^{\multimap} $ be the polynomial enumerating all graphs homomorphic to $\multimap$. Let $G_n$ be the polynomial enumerating all complete bipartite graphs. Then
\[G_n \leq f_n^{\multimap} \]
\end{lem}

\begin{proof}
$G_n$ can be written as
	$G_n\left( \left(y_e\right)_{e \in E(K_n)}\right) = \frac{1}{2} \sum_{V \subset [n]} \prod_{v \in V} \prod_{v' \in V^c} y_{v,v'}$. We express $G_n$ in the following way:
	\[G_n = \left (\sum_{k=1}^n \CHH[(n-k)k]{\bar x}{\CHH[k]{y}{\CHH[n-k]{z}{\CHH[1]{w}{f_{\multimap}^{n+2}(\bar x)}}}}\right )_{|_{w,y,z=1}}\]

Where $x_{n+1,n+2} = w$ and for all $i\in [n]$, $x_{i,n+1} = y$ and $x_{i, n+2} = z$.

Indeed, $G_n$ enumerates all bipartite graphs. Each operation above consisting in taking a homogeneous component eliminates some graphs from the list. We will demonstrate that, at the end, the polynomial will only enumerate complete bipartite graphs.

In a first step, we add two vertices $n+1$ and $n+2$ and we give a weight $w$ to the edge between them. By taking the homogeneous component of degree $1$ in the variable $w$, we eliminate all graphs but those which have an edge between $n+1$ and $n+2$.

Then we use variable substitutions to label the edges leaving $n+1$ with the weight $y$ and those leaving $n+2$ with $z$. A vertex linked to $n+1$ cannot be linked to $n+2$, as our graphs are bipartite and the vertices $n+1$ and $n+2$ are linked. Let us write $V$ and $V'$ for the partition of vertices defined by the bipartite graph, with $n+2 \in V$ and therefore, as $n+1$ and $n+2$ are linked, $n+1 \in V'$ 

By taking next the homogeneous component of degree $n-k$ in $z$, we force that $|V'| \geq n - k$. Similarly, by taking the homogeneous component of degree $k$ in $y$, we force that $|V| \geq k$. Therefore, as $V$ and $V'$ are disjoint sets, we force that $V^c = V'$.

At the end, by taking the homogeneous component of degree $(n-k)k$ in the variables $\bar x$, we only keep bipartite graphs built on $V$ and $V^c$, with $|V| = k$ and which have $(n-k)k$ edges, i.e., complete bipartite graphs on $n$ vertices.

The reduction is then given by Lemma~\ref{compo-homo}. 
\end{proof}

\subsection{Bipartite graphs}
Let us demonstrate that $F_n$ is $\VNP$-complete by reducing the generating function of clique to $F_n$. Recall that $\downarrow F_n = f^{\multimap}_n$ but that $F_n$ is not pure. 

\begin{lem}\label{bip-complet} Consider the polynomial
	\[F_n\left(\left(x_e\right)_{e \in E(K_n)}\right) = \text{GF}\left(K_n,\text{complete bipartite}\right) = \sum_{E'} \prod_{e \in E'} x_e,\]
where the sum is done over all subsets $E' \subset E(K_n)$ such that the graph $(n, E')$ (i.e., the graph with $n $ vertices and $E'$ as set of edges) has one complete bipartite connected component and all other connected components reduce to a single vertex.

This p-family is VNP-complete for c-reductions.
\end{lem}

\begin{proof}
We will reduce the clique generating function to our polynomial: $g_n(\bar x) = \text{GF}(K_n, \text{clique})((x_e)_{e \in E(K_n)}) = \sum_{E'}\prod_{e \in E'} x_e,$, where the summation is on all $E' \subset E(K_n)$ such that the graph $(n,E')$ has a complete connected component (a clique) and all others reduce to a single vertex. Bürgisser has demonstrated in~\cite{Bur00} that this sequence is VNP-complete for p-reductions.

Let us rewrite our polynomial:
\[F_n(\bar x) = \frac{1}{2} \sum_{\substack{V,V' \subset [n] \\ V \cap V' = \emptyset}} \prod_{v \in V} \prod_{v' \in V'} x_{(v,v')}\]

The $\frac{1}{2}$ comes from the fact that each complete bipartite graph can be defined by two different couples of sets: $(V,V')$ and $(V',V)$; and that each couple of vertex sets correspond to a unique complete bipartite graph.

We will proceed to the reduction by beginning with the sequence $(F_n)$ and writing a series of sequences of polynomials $(f^k_n)_{n \in \mathbb N}$, each being reducible to the previous one.

Hereafter, we suppose the left vertices of $K_{n,n}$ being $[1,n]$ and the right ones being $[n+1,2n]$. We write $\widehat i$ for the $i$-th right vertex, i.e., $\widehat i = i+n$. Similarly, if $[n]$ are the left vertices of $K_{n,n}$, we write $\widehat{[n]}$ for the right ones, i.e., $\widehat{[n]} = [n+1,2n]$.

\[f^1_n(\bar x) = F_{2n}(K_{n,n}) = \sum_{V \subset [n]} \sum_{V' \subset \widehat{[n]}} \prod_{i \in V} \prod_{j \in V'} x_{i,j}\]

Let $V$ and $V'$ be two disjoint sets of vertices. If $V$ has some of its vertices both in the right side and the left side (let us say for instance $i$ in the left side), then we can suppose that $V'$ has a vertex $j$ in the left side. Hence $x_{(i,j)} = 0$ appears in the weight of $V',V$. Therefore this weight is zero.

The only complete bipartite graphs enumerated in $F_{2n}(K_{n,n})$ are those built on $V$ and $V'$ with $V$ a subset of the left vertices and $V'$ a subset of the right ones. The summation on all the $V$ subset of the left vertices avoids repetitions and thus the $\frac{1}{2}$ coefficient is not necessary.

\[f^2_n(\bar x, y) = f^1_{n+1}(\bar x, x_{n+1,\widehat i} = 0, x_{i,\widehat{n+1}} = y) = \sum_{V \subset [n+1]} \sum_{V' \subset \widehat{[n+1]}} \prod_{i \in V} \prod_{j \in V'} x_{i,j}\]

If $V$ and $V'$ are vertex subsets, let $\omega(V,V')$ be their weight: $\omega(V,V') := \prod_{i \in V} \prod_{j\in V'} x_{i,j}$. Let $V \subset [n]$.
\begin{itemize}
	\item If $n+1 \in V$ then $\forall V' \subset \widehat{[n+1]}$, $\omega(V,V') = 0$.
	\item If $\widehat{n+1} \in V'$ then $\forall V \subset [n]$, $\omega(V,V') = (\prod_{i \in V} \prod_{j\in V'-\{\widehat{n+1}\}} x_{i,j}) \times y^{|V|}$.
	\item Else, for all $V \subset [n]$ $\text{deg}_y(\omega(V,V')) = 0$.
\end{itemize}
In those polynomials, the degree of $y$ is a witness to the size of $V$. We now just have to take the homogeneous component in $y$ of degree $k$ to keep all the $V$ of size $k$:

\[f^{3,k}_n(\bar x) = \CHH[k]{y}{f^2_n}(\bar x,y=1) = \sum_{\substack{V \subset [n] \\ |V| = k}} \sum_{V' \subset \widehat{[n]}} \omega(V,V')\]

Let us act similarly for $V'$.

\[f^{4,k}_n(\bar x, y) = f^{3,k}_{n+1}(\bar x, x_{n+1,\widehat{i}} = y, x_{i,\widehat{n+1}} = 0)\]

Here the degree of $y$ is a witness to the size of $V'$ in monomials where $y$ is of degree at least $1$. Then

\[f^{5,k}_n(\bar x) = \CHH[k]{y}{f^{4,k}_n}(\bar x,1) = \sum_{V \subset [n], |V| = k} \sum_{V' \subset \widehat{[n]}, |V'|=k} \omega(V,V')\]

Let us write now 
\[f^{6,k}_n(\bar x, y) = f^{5,k}_n(\bar x, x_{i,i'} = y) = \sum_{V \subset [n], |V| = k} \sum_{V' \subset \widehat{[n]}, |V'| = k} \prod_{i\in V} \prod_{\substack{j \in V'\\ j\neq i}} x_{i,j} \times y^{|V \cap V'|}\]

And at last, we find the clique generating function:
\[g_n(\bar x) = \sum_{k=1}^n \CHH[k]{y}{f^{6,k}_n}(\bar x,1)\]

Indeed, in this polynomial we sum on all subsets $V \subset [n]$ of size $k$ and the degree of $y$ is a witness to the size of $V \cap V'$. Then if we keep only monomials for which the degree in $y$ is $k$, we only keep complete bipartite graphs built on $V$ and $V'$, with $|V|=|V'|=|V \cap V'|$, i.e., on the same subset of vertices $V$, on the left and on the right. Then: 

\[\sum_{k+1}^n \CHH[k]{y}{f^{6,k}_n}(\bar x,1) = \sum_{V \subset [n]} \prod_{i \in V} \prod_{j \in V} x_{i,\widehat{j}} \]

This polynomial is, by definition, $ g_n(\bar x)$, the clique generating function. Furthermore, as we announce it at the beginning of this demonstration, we have the following reductions:

\[f^2_n \leq_p f^1_n \leq_p F_n\]

And, for all $ n $, $L^{f^2_{n+1}}(f^{3,k}_n) = \mathcal{O}(nk)$ then $L^{f^2_{n+1}}(f^{4,k}_n) = \mathcal{O}(nk)$. Similarly $L^{f^{4,k}_{n+1}}(f^{5,k}_n) = \mathcal{O}(nk)$ then $L^{f^2_{n+2}}(f^{6,k}_n) = \mathcal{O}(n^2k^2)$. At last $L^{f^2_{n+2}}(g_n) = \sum_{k=1}^n\mathcal{O}(n^3k^3) = \mathcal{O}(n^7)$.

Thus we have shown a c-reduction from $g_n$ to $F_n$.
\end{proof}

As we mentioned before the son of $F_n$ is $f_n^{\multimap} $ (i.e., $\downarrow F_n = f_n^{\multimap} $), but $F_n$ is not pure. We therefore had to introduce the polynomial $G_n$ which is still a generator of $f_n^{\multimap} $ but is pure. For now we introduce another polynomial, the cut enumerator. We will use it to prove the completeness of $G_n$. But as mentioned in the introduction, theorem~\ref{intro_Burgisser}, this family had been introduced by Bürgisser~\cite{Bur00} as a family neither $\VP$ nor $\VNP$-complete in $\mathbb{F}_2$, if of course $\VP \neq \VNP$.

Be aware that, contrary to the rest of this paper, this polynomial will be built on non symmetric variables: $x_{i,j} $ will, generally, not be equal to $ x_{j,i}$.

\begin{lem}\label{bip-complet-oriente}
 	Let $\cut_n$ be the following polynomial
	\[\cut_n\left(\left(x_{i,j}\right)_{i,j \in [n]}\right) = \sum_{V \subset [n]} \prod_{v \in V} \prod_{v' \in V^c} x_{v,v'}\]
This p-family is VNP-complete for c-reductions.
\end{lem}
 
\begin{proof}
We will exhibit a p-reduction from $F_n=\text{GF}(K_n, \text{complete bipartite})$ to $\cut_n$. Because $\text{GF}(K_n, \text{complete bipartite})$ is complete for c-reductions, as seen in Lemma~\ref{bip-complet}, this will prove the lemma.

Let us suppose that the left vertices of $K_{n,n}$ are $[n]$, the right ones $[n+1,2n]$ and let us write $\widehat i$ for the $i$-th right vertex, i.e., $\widehat i = i + n$. Similarly, if $[n]$ is the set of left vertices, we will write $\widehat{[n]}$ for the right ones. Let $(y_{i,j})_{i,j \in [2n]}$ be a set of new variables. We now give new values to the variables $x_{i,j} $ for all $ i,j \in [n]$:

\begin{enumerate}
	\item $ \forall i,j \in [n] $ $ x_{i,j} = 1$
	\item $ \forall i \in [n] $ $ x_{i,\widehat i} = 0$
	\item $ \forall i \in [n] $ $ x_{\widehat i,i} = 1$
	\item $ \forall i \in [n] \forall j \in \widehat{[n]} $ $ x_{j,i} = 1$
	\item $ \forall i \in \widehat{[n]} \forall j \in \widehat{[n]} $ $ x_{i,j} = 1$
	\item $ \forall i \in [n] \forall j \in \widehat{[n]} $ $ x_{i,j} = y_{i,j}$
\end{enumerate}

If we evaluate $\cut_{2n}$ on those variables, we compute $\text{GF}(K_n, \text{complete bipartite})$. Indeed, let $V \subset V(K_{n,n})$

\begin{itemize}
	\item If there exists  $i \in V$ such that $\widehat i \notin V$, then $\widehat i \in V^c$ and as $x_{i,\widehat i} = 0$ by (2), then $\omega(V) := \prod_{v \in V} \prod_{v' \in V^c} x_{v,v'} = 0$.
	\item If there exists $i \notin V$ such that $ \widehat i \in V$, then $i \in V^c $ and as $x_{\widehat i,i} = 1 $ by (3), the value of $\omega(V)$ is not necessarily zero.
\end{itemize}

$V$ is therefore of type $V_1 \cup V_2$ with $V_1$ in the left side, $V_2$ in the right one and $V_1' := \{\widehat n, n \in V_1\} \subset V_2$. Then if we write $V^c_1$ the complement of $V_1$ in the left side and $V^c_2$ the one of $V_2$ in the left side,

\begin{align*}
\omega(V) = & \prod_{v \in V} \prod_{v' \in V^c} x_{v,v'}\\
		  = & (\prod_{i \in V_1} \prod_{v' \in V^c} x_{i,v'})(\prod_{i' \in V_2} \prod_{v' \in V^c} x_{i',v'})
\\
		  = & (\prod_{i \in V_1} \prod_{j \in V_1^c} x_{i,j})(\prod_{i \in V_1} \prod_{j' \in V_2^c} x_{i,j'})(\prod_{i' \in V_2} \prod_{v' \in V^c} x_{i',v'}) \\
\end{align*} 

In the last line, the first parenthesis, $ (\prod_{i \in V_1} \prod_{j \in V_1^c} x_{i,j})$, takes value $1$ because of (1). The second one, $(\prod_{i \in V_1} \prod_{j \in V_2^c} x_{i,j})$, takes value $(\prod_{i \in V_1} \prod_{j \in V_2^c} y_{i,j})$ by (6).The last one, $ (\prod_{i \in V_2} \prod_{v' \in V^c} x_{i,v'}) $, take value $1$ because of (4) and (5). 

Then $\omega(V) = \prod_{i \in V_1} \prod_{j \in V_2^c} y_{i,j}$. But there is a bijection between  $\{V \subset V(K_{n,n}), \omega(V) \neq 0\}$ and $ \{(V,V'), V \subset [n], V' \subset [n]. V \cap V' = \emptyset\}$. Furthermore this bijection conserves weights. Thus, for the values of the variables $\bar x$ defined above, 

\[\cut_{2n}(\bar x) = \textsc{GF}(K_n, \text{complete bipartite})(\bar y)= F_n\]
\end{proof}

\begin{lem}\label{bip-complet-max}
	Let $G_n$ be the polynomial enumerating all complete bipartite graphs of size less than $n$. Then the p-family $(G_n)$ is VNP-complete for c-reductions.
\end{lem}

\begin{proof}
Classically, to pass from non oriented graphs to oriented graphs we double the number of vertices, $K_{2n}$, we call the $n$-th first as left vertices and the $n$-th last as right vertices. Then we labeled the edge between the $i$-th left vertex and the $j$-th right's by the oriented variable $x_{i,j}$. The edges between vertices of the same side are labeled with $1$ and horizontal edges, i.e., edges between the $i$-th left vertex and the $i$-th right's are labeled with $0$, as we do not want loops. Let us write $K'$ this new graph.

$G_n$ evaluate on this graph will enumerate every complete bipartite graphs. By taking homogeneous component we will reduce this list until we find the cut enumerator.

First, let us add to $K'$ two new vertices, $a$ and $b$. The new edge we create are labeled by:
\begin{itemize}
\item The edge between $a$ and $b$ is labeled by a new variable $t$,
\item The edge between $a$ and any left vertex by $y$,
\item The edge between $a$ and any right vertex by $y'$,
\item The edge between $b$ and any left vertex by $z$,
\item The edge between $b$ and any right vertex by $z'$,
\end{itemize}

By taking the homogeneous component of size $1$ in $t$, we only keep in our list the graphs which have a edge between $a$ and $b$.

For a $k \leq n$, we slip the left side of the graph $K'$ into two sets: $V$ and $V^c$ by joining $k$ vertices of the left side to $a$ and $n-k$ left vertices to $b$ (i.e., by taking the homogeneous component of size $k$ in $y$ and the homogeneous component of size $n-k$ in $z$). A vertex cannot be linked to both $a$ and $b$, as our graphs must always be bipartite. We name $V$ the set of left vertices linked to $a$, and therefore $V^c$ is the set of left vertices linked to $b$.

Similarly, we split the right side of the graph $K'$ into two sets: $V'$, the set of the $k$ right vertices linked to $a$ and $V'^c$ the set of the $n-k$ right vertices linked to $b$.

A vertex in $V$ cannot be linked to a vertex in $V'$ as they are both linked to $a$. Similarly a vertex in $V^C$ cannot be linked to a vertex in $V'^C$ because they are both linked to $b$. Furthermore, every edge between a vertex of $V$ and a vertex of $V^c$, or between a vertex of $V'$ and a vertex of $V'^c$ is labeled by $1$, as they are on the same side. Therefore, the only edge labeled with a $x_{i,j}$ variable that participate in $G_n$ are those between a vertex of $V$ and a vertex of $V'^c$ and those between one of $V'$ and one of $V'^c$.

If the $i$-th left vertex of $K'$ is in $V$ and the $i$-th right vertex of $K'$ is in $V'^c$, as $G_n$ enumerate only complete bipartite graphs, the monomial which represent this graph will have $x_{i,i} = 0$ and therefore will be null. Thus if we identify both the left and the right side to $[n]$, $V = V'$. Finally, if we evaluate $G_n$ on $K'$, and if evaluate the variables $t, y, y', z,$ and $ z'$ to $1$, we have

\[ \sum_{V \subset [n]}\left ( \prod_{i \in V} \prod_{j \in V'^c} x_{i,j} + \prod_{i \in V'} \prod_{j \in V^c} x_{i,j} \right)= 2 \cut(\bar x)\]
\end{proof}

\section*{Perspectives}
We have left open in section~\ref{hereditarity} the question of whether a pure multilinear polynomial is always reducible to its son, though we have shown the reduction for homogeneous polynomials and in the special case of  homomorphism polynomials. This question can be seen as a generalization of the  relationship between the permanent and the partial permanent $\per^*$ (cf~\cite{Bur00}). Note that the partial permanent is useful to show the completeness of other families of polynomials, so it would be interesting to obtain a general result.

The opposite question, namely whether the son of a polynomial always reduces to it, is also quite interesting. For instance, $\per$ is in $\VP$ in characteristic $2$. The same result for $\per^*$ was implicitly solved in~\cite{val01}. A reduction from $\downarrow f$ to f in the general case would provide another proof of this result and extend our understanding of these polynomials.

Another question raised by this paper is the relation between polynomials which list a graph property as homomorphic polynomials or generating functions on one hand and enumeration on the other. A brief foray into this subject has been made with corollary~\ref{enumeration}.

We intend to extend the study of this two questions in ours future researches.

\section*{Acknowledgements}
I thank to both of my doctoral advisors, A. Durand and G. Malod as well to L. Lyaudet for all his corrections.

\bibliographystyle{alpha}
\bibliography{biblio}

\end{document}